\documentclass[aps,pra,twocolumn,showpacs]{revtex4}
\usepackage{bm,bbm,amsmath, amssymb,amsbsy,graphicx,times,hyperref,amsthm,txfonts,subfigure}
\usepackage{color}

\renewcommand{\det}{{\rm Det}\,}

\newcommand{\be}{\begin{equation}}
\newcommand{\ee}{\end{equation}}
\newcommand{\bea}{\begin{eqnarray}}
\newcommand{\eea}{\end{eqnarray}}
\newcommand{\ket}[1]{|#1\rangle}
\newcommand{\bra}[1]{\langle#1|}

\newcommand{\eq}[1]{Eq.~(\ref{#1})}

\newtheorem{Theorem}{Theorem}

\begin{document}

\title{
Interplay between computable measures of entanglement and other quantum correlations
}

\author{Davide Girolami}
\email{pmxdg1@nottingham.ac.uk}
\author{Gerardo Adesso}
\email{gerardo.adesso@nottingham.ac.uk}
\affiliation{School of Mathematical Sciences, University of Nottingham, University Park, Nottingham NG7 2RD, United Kingdom}

\date{\today}

\begin{abstract}
Composite quantum systems can be in generic states characterized not only by entanglement, but also by more general quantum correlations. The interplay between these two signatures of nonclassicality is still not completely understood. In this work we investigate this issue focusing on computable and observable measures of such correlations: entanglement is quantified by the negativity ${\cal N}$, while general quantum correlations are measured by the (normalized) geometric quantum discord $D_G$. For two-qubit systems, we find that the geometric discord reduces to the squared negativity on pure states, while the relationship $D_G \ge {\cal N}^2$ holds for arbitrary mixed states. The latter result is rigorously extended to pure, Werner and isotropic states of two-qudit systems for arbitrary $d$, and numerical evidence of its validity for arbitrary states of a qubit and a qutrit is provided as well. Our results establish an interesting hierarchy, that we  conjecture to be universal, between two relevant and experimentally friendly nonclassicality indicators. This ties in with the intuition that general quantum correlations should at least contain and in general exceed entanglement on mixed states of composite quantum systems.
\end{abstract}


\maketitle

\section{Introduction}

The distinction between quantum and classical correlations in the state of a multipartite physical system is a fundamental problem with far-reaching implications \cite{OZ,HV}. Correlations can be regarded as genuinely  classical if they are essentially revealed by classical information theory tools as analogues of correlations between random variables. On the other hand,  entangled states have been traditionally considered as the only quantum-correlated class of states \cite{Werner89}, but this statement has recently proven to be misleading. Several features of separable (i.e. unentangled) states are incompatible with a purely classical description. To mention a few, an ensemble of separable nonorthogonal states cannot be discriminated perfectly \cite{nonlocwithoutent}, general separable states may have off-diagonal coherences in any product basis \cite{piani}, or, more practically, a measurement process on part of a composite system in a separable state may (and in general does) induce disturbance on the state of the complementary subsystems \cite{MID}. These are some genuine signatures of nonclassicality of correlations in the considered states, yet without any entanglement.

 A renewed attention towards the properties and the usefulness of such general quantum correlations in separable (and entangled) states has been triggered by the observation that, in mixed-state models of quantum computation (e.g. the so-called DQC1), such general quantum correlations may be at the root of a speed-up compared to classical scenario, despite the presence of zero or nearly vanishing entanglement \cite{dattabarbieriaustinchaves,dqc1discord}.
In general, it still remains an open issue whether such general quantum correlations are just related to the statistical properties of a state,  or represent truly some stronger, physical correlations of quantum nature that reduce to the entanglement in some cases, and go beyond it in general \cite{merali}.

It is clear that on pure bipartite states of arbitrary quantum systems, entanglement and quantum correlations are just synonyms. Both of them collapse onto the notion of lack of information about the system under scrutiny, when only a subsystem is probed. Quantitatively, this implies that any meaningful measure of entanglement or general quantum correlations should just reduce to some monotonic function of the marginal entropy of each reduced subsystem, when applied to pure bipartite states. The question becomes significantly more interesting for mixed bipartite states. One would expect to find, in general, an amount of quantum correlations that is no less than some valid entanglement monotone. In this paper we prove such an intuition to hold true for a particular choice of quantifiers of entanglement and quantum correlations, on arbitrary two-qubit states and on a relevant subclass of two-qudit states.

We recall that, in the last decade, a zoo of entanglement measures  (say ${\cal E}$ the amount of entanglement they aim to quantify) have been introduced \cite{PlenioVirmani}, and in a more recent drift several measures  have been proposed as well to evaluate the degree of general quantum correlations (say ${\cal Q}$) in composite systems \cite{OZ,HV,MID,altremisure,req,modi,dakic,genio}.  It seems reasonable to expect that
\begin{equation}
{\cal Q}  \geq  {\cal E}
\label{equestion}
 \end{equation}
 should hold for a {\it bona fide} chosen pair of quantifiers (see also \cite{nuovogenio}).
However, this claim turns out to be not mathematically fulfilled in some canonical cases. Selecting for instance two well-established entropic quantifiers such as the ``entanglement of formation'' \cite{eof} as an entanglement monotone, and the ``quantum discord'' \cite{OZ,HV} as a measure of quantum correlations, one finds that the latter can be greater as well as smaller than the former depending on the states, and no clear hierarchy can be established, even in the simple cases of two-qubit systems \cite{alqasimi} or two-mode Gaussian states \cite{AD_10}. An interesting study has recently succeeded in describing entanglement, classical and quantum correlations under a unified geometric picture \cite{modi}, by quantifying each type of correlations in terms of the smallest distance (according to the relative entropy) from the corresponding set of states without that type of correlations. For example, the amount of entanglement in a state $\rho$ is given by the relative-entropic distance between $\rho$ and its closest separable state, and it is called relative entropy of entanglement \cite{vedralrelent}. In this context, our expectation holds: the relative entropy of entanglement ${\cal E}_R$  is automatically smaller in general than the so-called relative entropy of quantumness ${\cal Q}_R$ \cite{req}, which in turn quantifies the minimum relative entropic-distance from the set of purely classically correlated states (a null-measure subset of the convex set of separable states \cite{acinferraro}). The latter measure ${\cal Q}_R$ has been recently interpreted operationally within an `activation' framework that recognizes the value of general quantum correlations as resources to generate entanglement with an ancillary system \cite{genio} (see also \cite{bruss,genioij}). Such a protocol is sufficiently general to let one define, in a natural way, quantumness measures ${\cal Q_E}$ associated to any proper entanglement monotone ${\cal E}$. In this way the question of the validity of \eq{equestion} becomes especially meaningful given the natural compatibility of the involved quantifiers \cite{nuovogenio}. However, there is a nontrivial optimization step required for the calculation of each ${\cal Q_E}$ that hinders the explicit computability of the desired resources.

 In this paper, we choose {\it computable} measures for entanglement and general quantum correlations. In the case of entanglement, we adopt the squared ``negativity'' ${\cal N}^2$ \cite{vidwer}, which is a measure of abstract algebraic origin, quantifying how much a bipartite state fails to satisfy the positivity of partial transpose (PPT) criterion for separability introduced by Peres and the Horodeckis \cite{PPT}. In the case of quantum correlations, we pick the ``geometric quantum discord'' $D_G$ \cite{dakic}, which measures (as suggested by the name) the minimum distance of a state from the set of classically correlated states, in terms of the squared Hilbert--Schmidt norm. Both measures are taken to be normalized between $0$ and $1$. Despite the very different origin and nature of these two measures, we prove that \eq{equestion} holds, namely $D_G \ge {\cal N}^2$, for arbitrary mixed states of two qubits.

We remark that both measures play key roles in the quantum correlation scenario, especially for their observability and usefulness in quantum information applications. In fact, the negativity is a popular entanglement measure, operationally related to the entanglement cost under PPT preserving operations \cite{pptcost}, and amenable of experimental estimation via quantitative entanglement witnesses (which provide measurable lower bounds to ${\cal N}$) \cite{quantwit}. On the other hand, the geometric discord, operationally interpreted in \cite{luofu}, also admits a tight lower bound $Q$ \cite{pirla} (which is by itself a faithful, observable quantifier of general quantum correlations), whose detection---not requiring complete state tomography---currently constitutes the optimal pathway to reveal and quantitatively estimate nonclassical correlations in quantum algorithms such as DQC1 mixed-state quantum computation \cite{dqc1discord}. In this respect, we show specifically that the chain $D_G\geq Q\geq {\cal N}^2$ holds on general two-qubit states (where the leftmost inequality is analytical \cite{pirla} and the rightmost one is corroborated by numerical simulations).

Furthermore, we prove that the inequality $D_G \ge {\cal N}^2$ extends to arbitrary pure, Werner \cite{Werner89} and isotropic states \cite{isot} of two qudits for any higher dimension $d$. We further provide numerical evidence that supports the validity of the inequality also in generic states of $2 \otimes 3$ systems. We then conjecture that $D_G \ge {\cal N}^2$ should hold for arbitrary mixed states of $d \otimes d'$ bipartite system. Our results demonstrate an interesting hierarchy between two apparently unrelated quantifiers of nonclassicality, for both of which closed formulas (and experimentally friendly detection schemes) are available on the classes of states considered here.

The fact that the geometric discord stands as a sharp upper bound on a computable measure of entanglement such as the (squared) negativity, is a worthwhile issue to impose a rigorous ordering of resources, for all those applications where the performance of a quantum information and communication primitive relies on the amount and the nature of nonclassical correlations between the involved parties \cite{dattacommun}.


 The paper is organized as follows. Sec.~\ref{SecMeas} recalls the definitions of negativity and geometric discord. In Sec.~\ref{SecQub} we compare the two measures on arbitrary states of two qubits. In Sec.~\ref{SecQud} we extend our analysis to higher-dimensional systems.  We summarize our results and discuss future perspectives in Sec.~\ref{SecConcl}.

\section{Measures of entanglement and quantum correlations}\label{SecMeas}

\subsection{Negativity}

According to the PPT  criterion \cite{PPT}, if a state $\rho_{AB} \equiv \rho$ of a bipartite quantum system is separable, then the partially transposed matrix $\rho^{t_A}$ is still a valid density operator, namely it is positive semidefinite. In general, $\rho^{t_A}$ is defined as the result of the transposition performed on only one ($A$, in this case) of the two subsystems in some given basis. Even though the resulting $\rho^{t_A}$ does depend on the choice of the
transposed subsystem and on the transposition basis, the statement $\rho^{t_A}\ge0$ is  invariant under such
choices \cite{PPT}. For $2 \otimes 2$ and $2 \otimes 3$ mixed states \cite{PPT}, for arbitrary $d \otimes d'$ pure states, and for all Gaussian states of $1 \otimes n$ mode continuous variable systems \cite{Simonwerwolf}, the PPT criterion is a necessary and sufficient condition for separability and, at the same time, its failure reliably marks the presence of entanglement. In all the other cases, there exist states which can be entangled yet with a positive partial transpose: they are so-called bound entangled states, whose entanglement cannot be distilled by means of local operations and classical communications (LOCC) \cite{bound}.

On a quantitative level, the negativity of the partial transpose, or, simply, ``negativity'' ${\cal N}(\rho)$ \cite{zircone,vidwer} can be adopted as a valid, computable measure of (distillable) entanglement for arbitrary bipartite systems. The negativity of a quantum state $\rho$ of a bipartite $d \otimes d$ system can be defined as
\begin{equation}\label{nega}
{\cal N}(\rho) = \frac{1}{d-1} (\|\rho^{t_A}\|_1 -1)\,,
\end{equation}
where
\begin{equation}\label{norm1}
\|M\|_1=\text{Tr}|M| = \sum_i |m_i|\
\end{equation}
stands for the $1$-norm, or trace norm, of the matrix $M$ with eigenvalues $\{m_i\}$.
The quantity ${\cal N} (\rho)$ is proportional to the modulus of the sum of the negative
eigenvalues of $\rho^{t_A}$, quantifying the extent to which
the partial transpose fails to be positive.

 The negativity ${\cal N}$ is in general an easily computable entanglement measure, and it has been proven to be (along with its square ${\cal N}^2$)  convex and monotone under LOCC \cite{vidwer}. The squared negativity ${\cal N}^2$ satisfies a monogamy inequality on the sharing of entanglement for multiqubit systems \cite{ckwyongche}.

\subsection{Geometric quantum discord}

The ``geometric quantum discord'' $D_G$ has been recently introduced as a simple geometrical quantifier of general nonclassical correlations in bipartite quantum states \cite{dakic}. Let us suppose to have a bipartite system $AB$ in a state $\rho$ and to perform a local measurement on the subsystem $B$. Almost all (entangled or separable) states will be subject to some disturbance due to such a measurement \cite{acinferraro}. However, there is a subclass of states which is left unperturbed by at least one measurement: it is the class of the so-called ``classical-quantum'' states \cite{piani}, whose representatives have a density matrix of this form
\begin{eqnarray}\label{cq}
\rho = \sum _i  p_i \rho _{Ai}\otimes |i\rangle  \langle i | ,
\end{eqnarray}
where $p_i$ is a probability distribution, $\rho_{Ai}$ is the marginal density matrix of $A$ and  $\{|i\rangle\}$ is an orthonormal vector set.
 Letting $\Omega$ be the set of classical-quantum states,  and $\chi$ be a generic element of this set, the geometric discord $D_G$ is defined as  the distance between the state $\rho$ and the closest classical-quantum state. In the original definition \cite{dakic}, the (unnormalized) squared Hilbert--Schmidt distance is adopted. We employ here a normalized version of the geometric quantum discord for arbitrary mixed states $\rho$ of a  $d \otimes d$ quantum system,
 \begin{equation}\label{dgeom}
  D_G(\rho)=\frac{d}{d-1} \min_{\chi \in \Omega}  \|\rho -\chi \|_2^2\,.
  \end{equation}
where
\begin{equation}\label{norm2}
\|M\|_2=\sqrt{\text{Tr}(M M^\dagger)} = \sqrt{\sum_i m_i^2}\,,
\end{equation}
stands for the $2$-norm, or Hilbert--Schmidt norm, of the matrix $M$ with eigenvalues $\{m_i\}$.
The quantity $D_G (\rho)$ in \eq{dgeom} is normalized between $0$ (on classical-quantum states) and $1$ (on maximally entangled states $\rho=\ket{\psi}\!\bra{\psi}$, $\ket{\psi} = d^{-1/2} \sum_{j=0}^{d-1} \ket{j}\ket{j}$).

The geometric discord can be re-interpreted as the minimal disturbance, again measured according to the squared Hilbert--Schmidt distance, induced by any projective measurement $\Pi^B$ on subsystem $B$ \cite{luofu},
 $$ D_G(\rho)=\frac{d}{d-1} \min_{\Pi^B}  \|\rho -\Pi^B(\rho) \|_2^2\,.$$
We notice that the geometric discord is not symmetric under a swap of the two parties, $A \leftrightarrow B$.

The minimization involved in the definition of the geometric quantum discord can be solved exactly for arbitrary two-qubit states \cite{dakic} and pure two-qudit states \cite{luofu,luoMIL}, leading to computable formulas, as detailed in the following Sections. In the remainder of the paper, we will compare entanglement --- quantified by ${\cal N}^2$ --- and quantum correlations --- quantified by $D_G$. The latter will be shown to majorize the former.  We observe that picking the square of the negativity as entanglement measure is unconventional, yet necessary in this case: we want  to make a mathematically consistent comparison of measures, both acting quadratically on the eigenvalues of the involved matrices [compare Eqs.~(\ref{nega}) and (\ref{dgeom})].

\section{Geometric discord versus negativity in two-qubit systems}\label{SecQub}

The main result of this Section is the following.

  \begin{Theorem}\label{teo1}For every  general two-qubit state $\rho$, the geometric quantum discord is always greater or equal than the squared negativity,
  \begin{equation}\label{tesi}
  D_G(\rho) \geq {\cal N}^2(\rho)\,.
  \end{equation}
  \end{Theorem}

Let us review the formulas needed to evaluate the two chosen measures for generic two-qubit states.

 The geometric discord $D_G$ admits  an explicit closed expression for two-qubit states \cite{dakic}. First, one needs to express the $4\times 4$ density matrix $\rho$ of a two-qubit state in the so-called Bloch basis (or $R$-picture) \cite{verstraete}:
 \begin{eqnarray}\label{bloch}
  \rho&=& \frac14 \sum_{i,j=0}^3 R_{ij} \sigma_i \otimes \sigma_j \\
  \nonumber &=& \frac 14\bigg(\mathbb{I}_{4}+\sum_{i=1}^3 x_i\sigma_i \otimes \mathbb{I}_{2} \\ \nonumber
  & & \quad +\sum_{j=1}^3 y_j \mathbb{I}_{2}\otimes \sigma_j+\sum_{i,j=1}^3 t_{ij} \sigma _i\otimes\sigma_j\bigg),
  \end{eqnarray}
where $R_{ij}=\text{Tr}[\rho(\sigma_i\otimes \sigma_j)]$, $\sigma_0=\mathbb{I}_{2}$, $\sigma _i$ ($i=1,2,3$) are the Pauli matrices, $\vec{x}=\{x_i\},\vec{y}=\{y_i\}$ are the three-dimensional Bloch column vectors associated to the subsystems $A,B$, and  $t_{ij}$ denote the elements of the correlation matrix $T$.  Then, following \cite{dakic}, the normalized geometric discord $D_G$, \eq{dgeom}, takes the form
 \begin{equation}\label{dgformula}
 D_G(\rho)= \frac 12(\|\vec y\|^2 + \|T\|_2^2 -k),
  \end{equation}
  with $k$ being the largest eigenvalue of the matrix $\vec y {\vec  y}^t+  T^t T$.
The expression in \eq{dgformula} can be also recast as the solution to a variational problem  \cite{luofu}; namely,  for two qubits,
 \begin{equation}\label{dgformula2}
 D_G(\rho)= 2 \left[\text{Tr}(C^t C) - \max_A \text{Tr}(A C^t C A^t)\right]\,,
  \end{equation}
where $C=R/2$ and the maximum is taken over all $2\times 4$ isometries $A = \frac{1}{\sqrt2}\left(
                                                                                              \begin{array}{cc}
                                                                                                1 & \vec a \\
                                                                                                1 & -\vec a \\
                                                                                              \end{array}
                                                                                            \right)
$, with $\vec a$ a three-dimensional unit vector.

Concerning the negativity ${\cal N}$, \eq{nega}, it is known that a two-qubit state $\rho$ is separable if and only if $\rho^{t_A} \ge 0$ \cite{PPT}, and, for entangled two-qubit states $\rho$, at most one eigenvalue of the partial transpose $\rho^{t_A}$ can be negative \cite{verstraete}. Denoting by $\{\lambda_i\}$  the eigenvalues of $\rho^{t_A}$ in decreasing order, for two-qubit entangled states we have $\lambda_1 \ge \lambda_2 \ge \lambda_3 \ge 0\ge \lambda_4$ and  the negativity of $\rho$ takes the form \cite{vidwer}
\begin{equation}\label{negaformula}
{\cal N}(\rho)= \|\rho ^{t_A}\|_1-1=2 |\lambda_4| \,,
\end{equation}
 while for separable states ($\lambda_4 \ge 0$) one has
${\cal N}(\rho)=0$.

We first compare entanglement and quantum correlations in the simple instance of pure two-qubit states $\rho^p=\ket{\psi}\!\bra{\psi}$. Up to local unitary operations (which leave correlations invariant), a two-qubit pure state can be written in its Schmidt decomposition, corresponding to a density matrix of the form
\begin{equation}\label{rhopure}
\rho^p=
\left(
\begin{array}{cccc}
 \frac{1}{2} \left(\sqrt{1-\mathcal{N}^2}+1\right) & 0 & 0 & \frac{\mathcal{N}}{2} \\
 0 & 0 & 0 & 0 \\
 0 & 0 & 0 & 0 \\
 \frac{\mathcal{N}}{2} & 0 & 0 & \frac{1}{2} \left(1-\sqrt{1-\mathcal{N}^2}\right)
\end{array}
\right)
\end{equation}
It is straightforward to show that in this case,
\begin{equation}
D_G(\rho^p)={\cal N}^2(\rho^p) \equiv S_L(\rho^p_A)\,,\label{pured}
\end{equation}
where $S_L(\rho^p_A)=4 \det(\rho^p_A)$ denotes the marginal linear entropy of one subsystem in its reduced state. As expected, entanglement and quantum correlations correctly coincide for pure two-qubit states, and specifically the two chosen measures (geometric discord and squared negativity) collapse onto the very same quantifier of local lack of purity.

For general two-qubit mixed states,  our intuition dictates that the amount of quantum correlations should exceed entanglement. This is  formalized in Theorem~\ref{teo1}, which we are now ready to prove.

  \begin{proof}
  We focus on the case of entangled states, as \eq{tesi} trivially holds when $\rho$ is separable.  \\
     First,   we have a look at the original formulation of geometric discord in \cite{dakic}: the closest classical-quantum state $\bar{\chi}$ that achieves the minimum of the Hilbert--Schmidt norm $||\rho - \chi||_2^2$ is  such that  $\text{Tr}[\rho{\bar{\chi}}]=\text{Tr}[{\bar{\chi}}^2]$. Thus, we can rewrite \eq{dgeom} as
\begin{eqnarray}\label{semp}
D_G&=&2\ \min_{ \chi \in \Omega}\|\rho - \chi\|_2^2 = 2\left(\text{Tr}[\rho^2]-\text{Tr}[{\bar{\chi}}^2]\right)\nonumber\\
&=&2\left(\text{Tr}[\rho^{t_A\ 2}]-\text{Tr}[{\bar{\chi}}^2]\right).
\end{eqnarray}
 Then, denoting (as before) by  ${\bf \lambda}=\{\lambda_i\}$ the vector of eigenvalues of  $\rho^{t_A}$ in decreasing order ($\lambda_1 \geq\lambda_2\geq\lambda_3\geq 0\geq \lambda_4$), and similarly denoting by ${\bf \varsigma}=\{\varsigma_i\}$ the vector of eigenvalues of ${\bar{\chi}}$ ($\varsigma_1 \geq\varsigma_2\geq\varsigma_3\geq \varsigma_4\geq 0$), recalling that  the Hilbert--Schmidt norm is invariant under partial transposition \cite{geover}, we obtain $\sum_{i=1}^4\varsigma_i^2= \text{Tr}[\rho^{t_A}{\bar{\chi}}]$. We can further exploit the Hoffman--Wielandt theorem \cite{horn}, which implies that
 \begin{equation}\label{resparz}
 \|\rho^{t_A}-{\bar{\chi}}\|_2^2 \geq \sum_{i=1}^4|\lambda_i-\varsigma_i|^2 =\sum_{i=1}^3 |\lambda_i-\varsigma_i|^2 + (|\lambda_4|+\varsigma_4)^2.
\end{equation}
   Thus, from (\ref{semp}) and (\ref{resparz})  we have
   \begin{eqnarray}\label{ermetodo}
   \sum_{i=1}^4 \varsigma_i^2 = \sum_{i=1}^4  \lambda_i\varsigma_i.
   \end{eqnarray}
    Now, let  us consider the function
   \begin{eqnarray}
   f({\bf \lambda},{\bf \varsigma})=\sum_{i=1}^3\lambda_i|\lambda_i-\varsigma_i| -|\lambda_4|(|\lambda_4|-\varsigma_4);
   \end{eqnarray}
   it is easy to see  that, performing an optimization by the Lagrange multipliers method, the minimum of $f$ keeping fixed $|\lambda_4|$ and $\varsigma_4$ (say $f'$) is reached when $\lambda_1=\lambda_2=\lambda_3=\frac{(1+|\lambda_4|)}3$ and $\varsigma_1=\varsigma_2=\varsigma_3=\frac{1-\varsigma_4}{3}$.  Hence, we have
\[
   f'(|\lambda_4|,\chi_4)=(1+|\lambda_4|)\left(\frac{1+|\lambda_4|}3-\frac{1-\varsigma_4}3\right)-|\lambda_4|(|\lambda_4|-\varsigma_4).
  \]
   Furthermore, optimizing over  $\varsigma_4$ we obtain $f''$, which is the minimum of $f$ at fixed $|\lambda_4|$ (i.e. at fixed negativity):
    \begin{eqnarray}
   f''(|\lambda_4|)= \frac{1+|\lambda_4|}3 -|\lambda_4|\geq 0.
   \end{eqnarray}
Finally, the last inequality implies
$\sum_{i=1}^3\lambda_i|\lambda_i-\varsigma_i| \geq |\lambda_4|(|\lambda_4|-\varsigma_4)$, i.e.,
\[
 \sum_{i=1}^3\lambda_i|\lambda_i-\varsigma_i| +|\lambda_4|(|\lambda_4|+\varsigma_4)\geq 2|\lambda_4|^2,
 \]
and thanks to \eq{ermetodo} this yields
\begin{equation}
\sum_{i=1}^4|\lambda_i-\varsigma_i|^2 \geq  2|\lambda_4|^2\,,
\end{equation}
which is equivalent to \eq{tesi}, thus demonstrating the claim.
This concludes the proof of Theorem \ref{teo1} for all two-qubit mixed states.
  \end{proof}

To illustrate the comparison between geometric discord and squared negativity, we plot in Fig.~\ref{figqubits} the physical region filled by $10^5$ randomly generated two-qubit states in the space  $D_G$ versus ${\cal N}^2$. Along with the (red online) lower bound emerging from Theorem \ref{teo1},  saturated by pure states [\eq{rhopure}] for which $D_G = {\cal N}^2$, we notice the existence of an upper bound as well on $D_G$ at fixed negativity. This shows that the quantum correlations in excess of entanglement or, in general, beyond entanglement are somehow constrained. Two-qubit states saturating the (green online) upper bound can be sought within the class of rank-two $X$-shaped density matrices of the form
\begin{equation}
\label{X2}
\rho^X=\left(
\begin{array}{cccc}
 a & 0 & 0 & \sqrt{a d} \\
 0 & b & \sqrt{b c} & 0 \\
 0 & \sqrt{b c} & c & 0 \\
 \sqrt{a d} & 0 & 0 & d
\end{array}
\right)\,,
\end{equation}
where $d=1-a-b-c$ and $b=\big[2-2 a-2 c+2 \big(-1+6 a-7 a^2+6 c-18 a c-7 c^2+4 \sqrt{2} \sqrt{a c (-1+2 a+2 c)^2}\big)^{\frac12}\big]/4$, with $a$ and $c$
varying in the parameter range $0 \le a,c \le 1/2,\, -1+6 a-7 a^2+6 c-18 a c-7 c^2+4 \sqrt{2} \sqrt{a c} |2 a+2 c-1| \ge 0$. The remaining optimization of $D_G$ at fixed ${\cal N}^2$ can be efficiently done numerically.

\begin{figure*}[t]
\subfigure[]{\includegraphics[width=8.2cm]{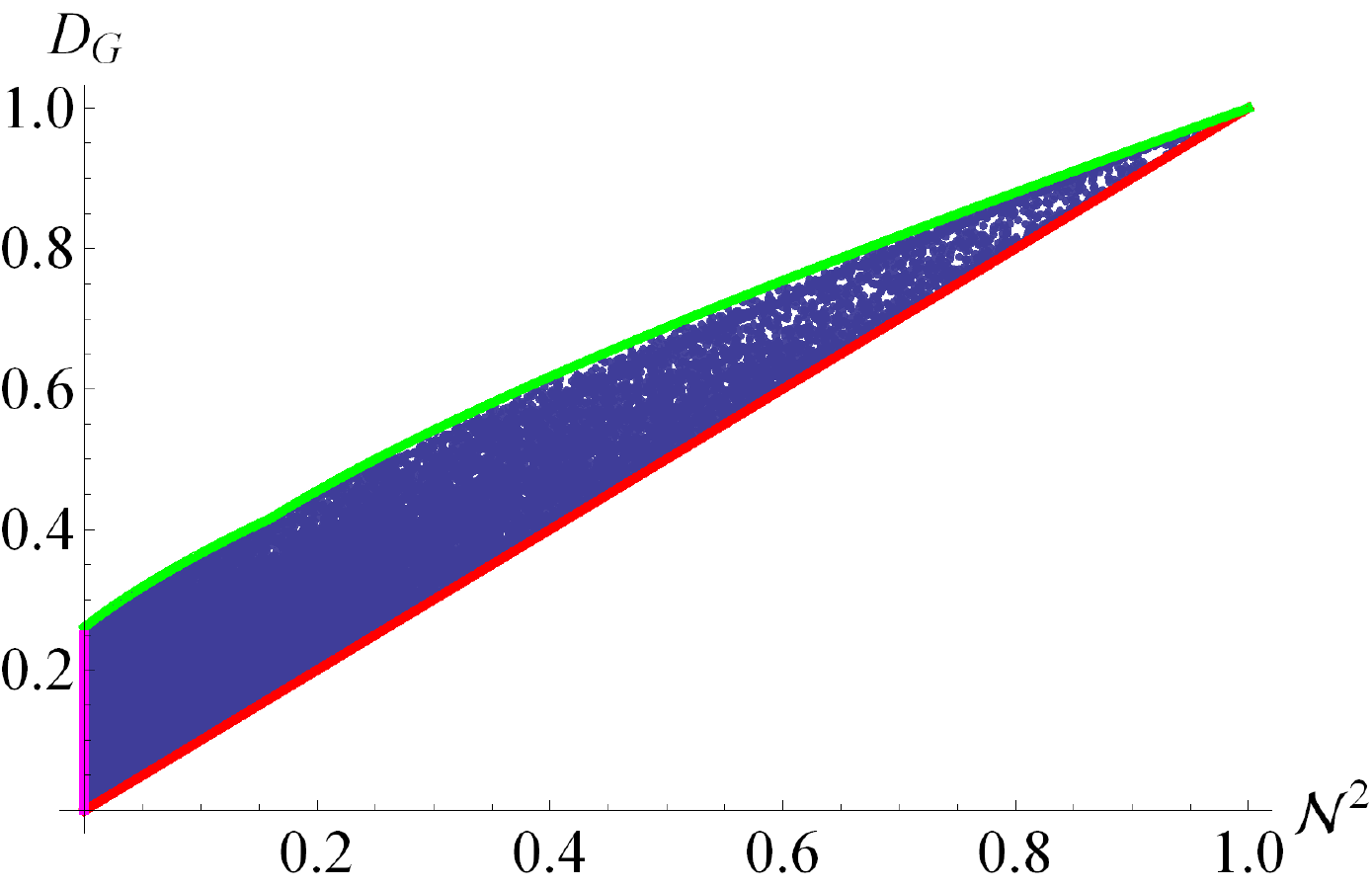}\label{figqubits}}\hspace*{1.3cm}
\subfigure[]{\includegraphics[width=8.2cm]{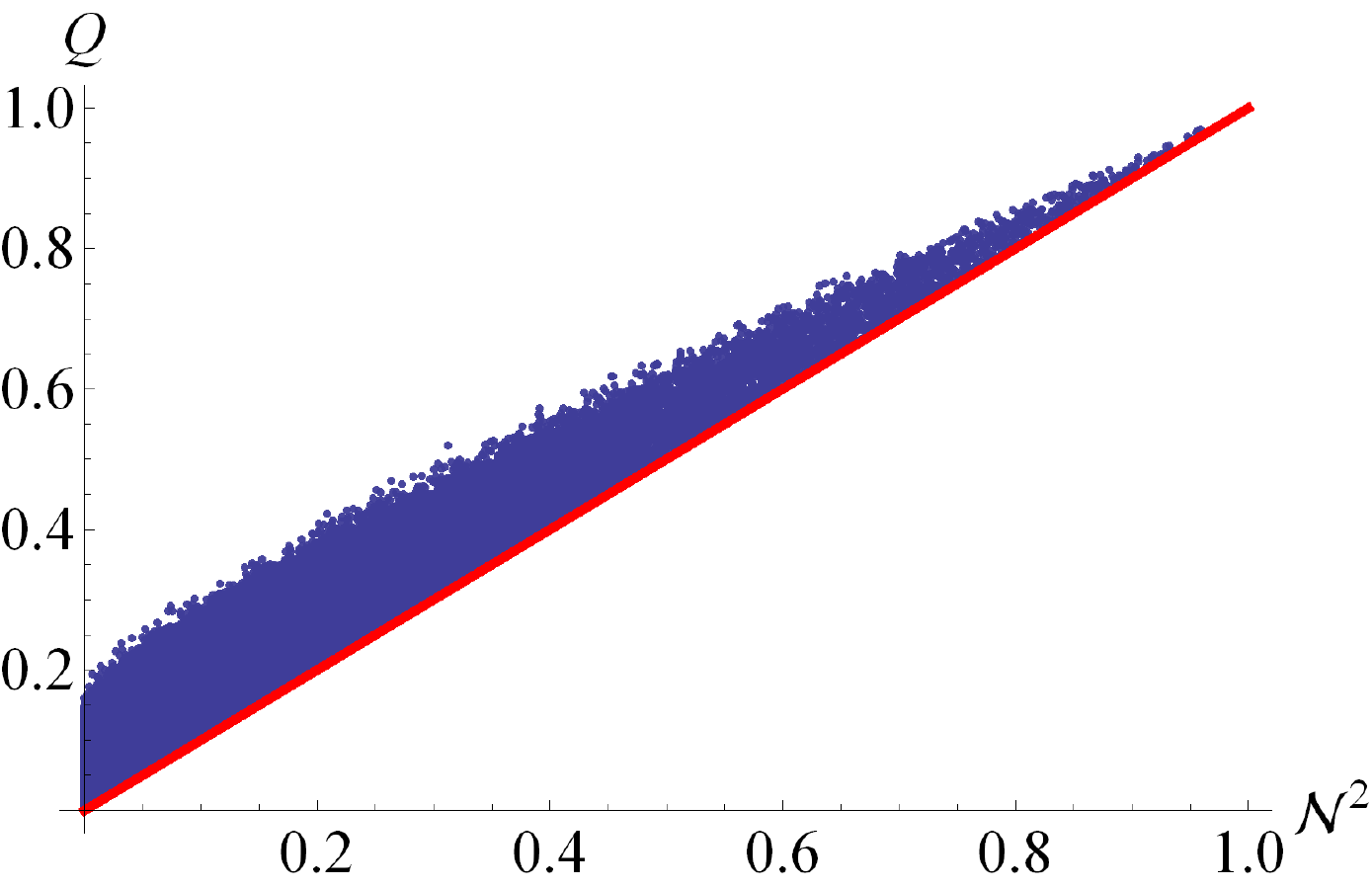}\label{qvsneg}}
\caption{(Color online) Geometric quantum discord $D_G$ (a) and its observable lower bound $Q$ (b) versus squared negativity ${\cal N}^2$ for $10^5$ randomly generated states of two qubits. The lower boundary (red online) in both plots accommodates pure states. In panel (a), the upper boundary (green online) can be saturated by a subclass of rank-two states of the form \eq{X2}; while the side (magenta online) vertical line at ${\cal N}^2=0$ is filled by separable states yet with nonzero quantum correlations, which reach up to the value $D_G = 1/4$ on states of the form \eq{sepopt}. All the quantities plotted are dimensionless. }\label{fig1ab}
\end{figure*}

In the limiting case of separable two-qubit states, ${\cal N}(\rho^{sep})=0$, the maximum value of the (normalized) geometric discord can be analytically found to be  \cite{notevlatko} \begin{equation}
D_G(\rho^{sep}_{opt})=\frac14\,.
 \end{equation}
 This is achieved by imposing the edge of separability, $\lambda_4=0$, that corresponds to $a d = b c$ in \eq{X2}. The maximum $D_G$ is then reached e.g. for $a=c=\frac{1}{8} \left(2+\sqrt{2}\right)$. Notice that the corresponding state $\rho^{sep}_{opt}$,
 \begin{equation}
 \label{sepopt}
\rho^{sep}_{opt}= \left(
\begin{array}{cccc}
 \frac{1}{8} \left(2+\sqrt{2}\right) & 0 & 0 & \frac{1}{4 \sqrt{2}} \\
 0 & \frac{1}{8} \left(2-\sqrt{2}\right) & \frac{1}{4 \sqrt{2}} & 0 \\
 0 & \frac{1}{4 \sqrt{2}} & \frac{1}{8} \left(2+\sqrt{2}\right) & 0 \\
 \frac{1}{4 \sqrt{2}} & 0 & 0 & \frac{1}{8} \left(2-\sqrt{2}\right)
\end{array}
\right),
\end{equation}
 upon swapping the subsystems $A$ and $B$, becomes of the classical-quantum form of \eq{cq}, i.e., a state with zero $D_G$. This suggests that the maximum geometric discord for general two-qubit separable states, is obtained on an extremally asymmetric state (the marginal state ${\rho^{sep}_{opt}}_A$ is maximally mixed, while the marginal state of subsystem $B$ is quasi-pure, $\text{Tr}\big({{\rho^{sep}_{opt}}_B}^2\big)=3/4$), that displays no signature of quantum correlations at all if subsystem $A$ rather than $B$ is probed by local measurements. The example in \eq{sepopt} is just one of an entire class of two-qubit states that enjoy the same property \cite{genioij}.

 The full allowed range $0 \le D_G \le 1/4$ for the geometric discord of separable states (magenta line in Fig.~\ref{figqubits}) can be spanned for instance by mixtures of the form $\rho^{sep}_p = p \rho^{sep}_{opt} + (1-p) I/4$, with $0 \le p \le 1$, for which $D_G(\rho^{sep}_p)=p^2/4$.

We can refine the hierarchy proven in this Section by taking into account the observable measure of quantum correlations $Q$ introduced in \cite{pirla}. In particular, for arbitrary two-qubit states this quantity takes the form of a state-independent function of the density matrix elements given by
\begin{eqnarray}
Q=\frac 23\left(\text{Tr}[S]-\sqrt{6\text{Tr}[S^2]-2(\text{Tr}[S])^2}\right),
\end{eqnarray}
where $S=\frac 14( \vec{y}\vec{y}^t+T^tT)$. We have shown in \cite{pirla} how to recast $Q$ in terms of observables that can be measured experimentally via simple quantum circuits. We also proved that $Q$ is a tight lower bound to the geometric discord, i.e. $D_G\geq Q$, where the inequality is saturated for pure states and $Q=0\iff D_G =0$. In Fig.~\ref{qvsneg} we plot $Q$ versus the squared negativity: numerics confirm that this novel quantity is still an upper bound to ${\cal N}^2$. Therefore, the following hierarchical ordering is satisfied for all two-qubit states: $D_G \geq Q\geq {\cal N}^2$, while all the quantifiers become equal for pure states.

\section{Geometric discord versus negativity in higher-dimensional systems}\label{SecQud}

Here we provide extensions of the results of the previous Section to $d \otimes d$ and $d \otimes d'$ systems.

\subsection{Pure $\boldsymbol{d \otimes d}$ states}

\begin{figure*}[t]
\includegraphics[width=16cm]{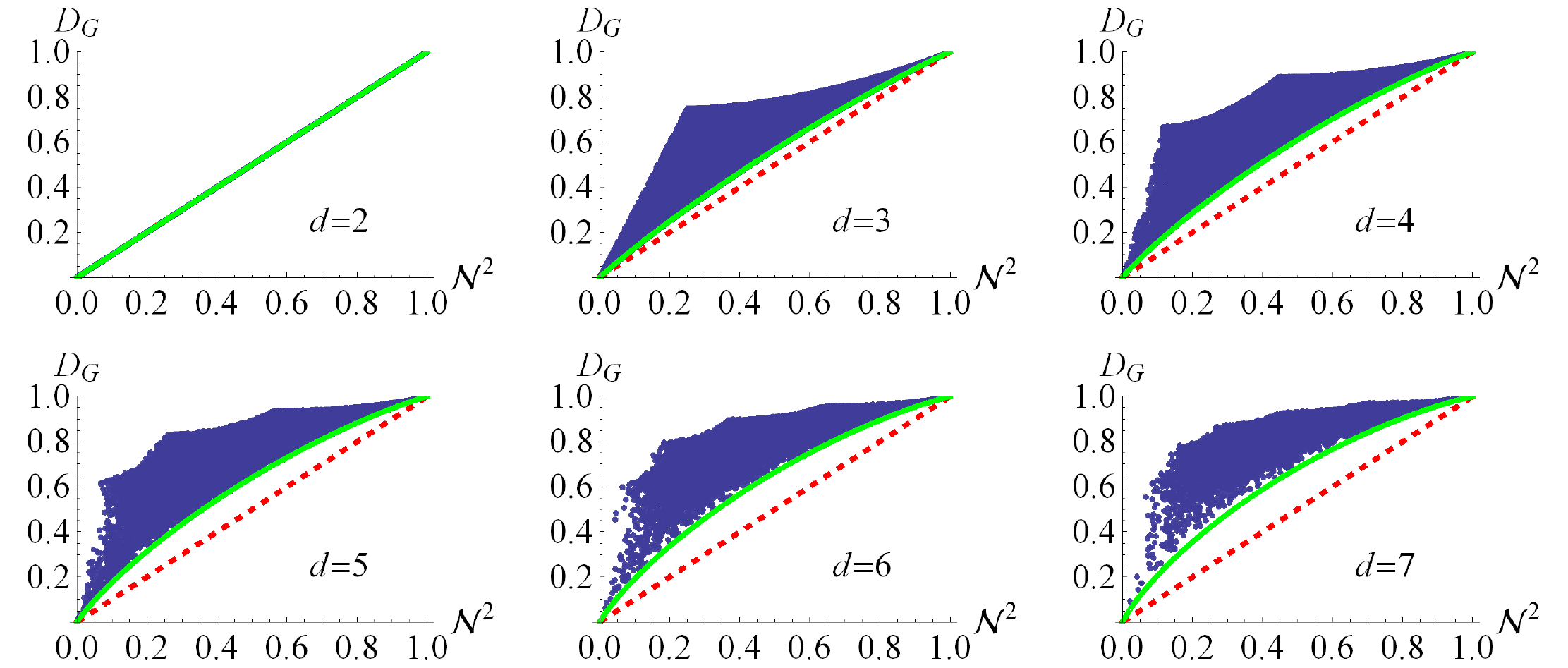}
\caption{(Color online) Geometric quantum discord versus squared negativity for $3 \times 10^4$ (per panel) randomly generated states pure states of two qudits with $d=2,\ldots,7$. The two measures coincide for $d=2$ (pure two-qubit states). In general, the line (dashed red online) $D_G={\cal N}^2$ is not attainable for intermediate values of both measures, while a tighter lower bound (solid green online) on $D_G$ exists at fixed negativity, given by \eq{lboundd}. Such a bound is saturated by states with Schmidt decomposition as in \eq{saturi}. The upper bound on $D_G$ at fixed negativity is more structured. Notice that the plots in this Figure can be also interpreted as the span of the pair of entanglement measures $\tau_2$ \cite{zik2} versus ${\cal N}^2$ \cite{vidwer} for two-qudit pure states.   All the quantities plotted are dimensionless. }\label{figqudits}
\end{figure*}

We first generalize Theorem \ref{teo1} to arbitrary pure states of two qudits. Namely, we prove the following.

  \begin{Theorem}\label{teo2}For every  pure two-qudit state $\ket{\psi} \in {\mathbb{C}^d \otimes \mathbb{C}^d}$, the geometric quantum discord is always greater or equal than the squared negativity,
  \begin{equation}\label{tesi2d}
  D_G(\psi) \geq {\cal N}^2(\psi)\,.
  \end{equation}
  \end{Theorem}

 \begin{proof}
 Any pure state $\ket{\psi} \in \mathbb{C}^d \otimes \mathbb{C}^d$ can be written without loss of generality in the Schmidt decomposition
\begin{equation}
\label{schd}
{\ket{\psi}} = \sum_{j=0}^{d-1} \sqrt{\alpha_j} \ket{j}\ket{j}\,,
\end{equation}
where the Schmidt coefficients are probability amplitudes, $\sum_j \alpha_j = 1$.

The geometric discord [\eq{dgeom}] can be computed in this case following Luo and Fu \cite{luofu,luoMIL}. The closest classical state to $\ket{\psi}$, entering the definition [\eq{dgeom}], turns out to be the completely uncorrelated state $\rho^{\otimes} = \rho_A \otimes \rho_B$, obtained as the tensor product of the marginal states $\rho_{A} = \text{Tr}_{B} (\ket{\psi}\!\bra{\psi})$ and $\rho_{B} = \text{Tr}_{A} (\ket{\psi}\!\bra{\psi})$. This implies
\begin{equation}\label{dgd}
D_G(\psi) = \frac{d}{d-1}\left(1-\sum_i \alpha_i^2\right) = \frac{2d}{d-1} \sum_{j>i}\alpha_i\alpha_j\,.
\end{equation}

Meanwhile, the negativity [\eq{nega}] is given by \cite{vidwer}
   \begin{eqnarray}\label{nd}
   {\cal N}(\psi)&=&\frac{1}{d-1}\left[\left(\sum_i \sqrt{\alpha_i}\right)^2 -\sum_i \alpha_i\right]\nonumber\\
   &=&\frac{1}{d-1}\left[\left(\sum_i \sqrt{\alpha_i}\right)^2-1\right].
   \end{eqnarray}
   We know from \cite{qudits} that the following inequality holds:
   \begin{eqnarray}
  4\sum_{j>i}\alpha_i\alpha_j \geq  \frac 2{d(d-1)}\left[\left(\sum_i \sqrt{\alpha_i}\right)^2-1\right]^2,
   \end{eqnarray}
therefore we obtain
  \begin{eqnarray}
  2 \frac{d}{d-1}\sum_{j>i}\alpha_i\alpha_j \geq  \frac 1{(d-1)^2}\left[\left(\sum_i \sqrt{\alpha_i}\right)^2-1\right]^2.
   \end{eqnarray}
The left side is the normalized geometric discord, while on the right we have the normalized squared negativity.
  \end{proof}

We have already seen that for $d=2$, the two measures $D_G$ and ${\cal N}^2$ indeed coincide on pure states. However, for any $d>2$, the geometric discord is in general strictly larger than the negativity. This seems to go against the expectation that quantum correlations should reduce to entanglement on pure states. In fact, $D_G$ does reduce to an entanglement measure on general two-qudit pure states, but such a measure is in general different from the squared negativity for $d \ge 3$. The pure-state entanglement monotone that takes the very same expression as in \eq{dgd} is a particular coefficient $\tau_2$ of  the characteristic polynomial of the nontrivial block of the Gram matrix of pure two-qudit states (see \cite{zik2} for details). Such a measure has not been studied for mixed states, and it is an interesting (yet technically challenging) open problem to see whether  the hierarchy $D_G \ge \tau_2$ holds for general two-qudit mixed states beyond $d=2$.

\begin{figure*}[t]
\includegraphics[width=17.8cm]{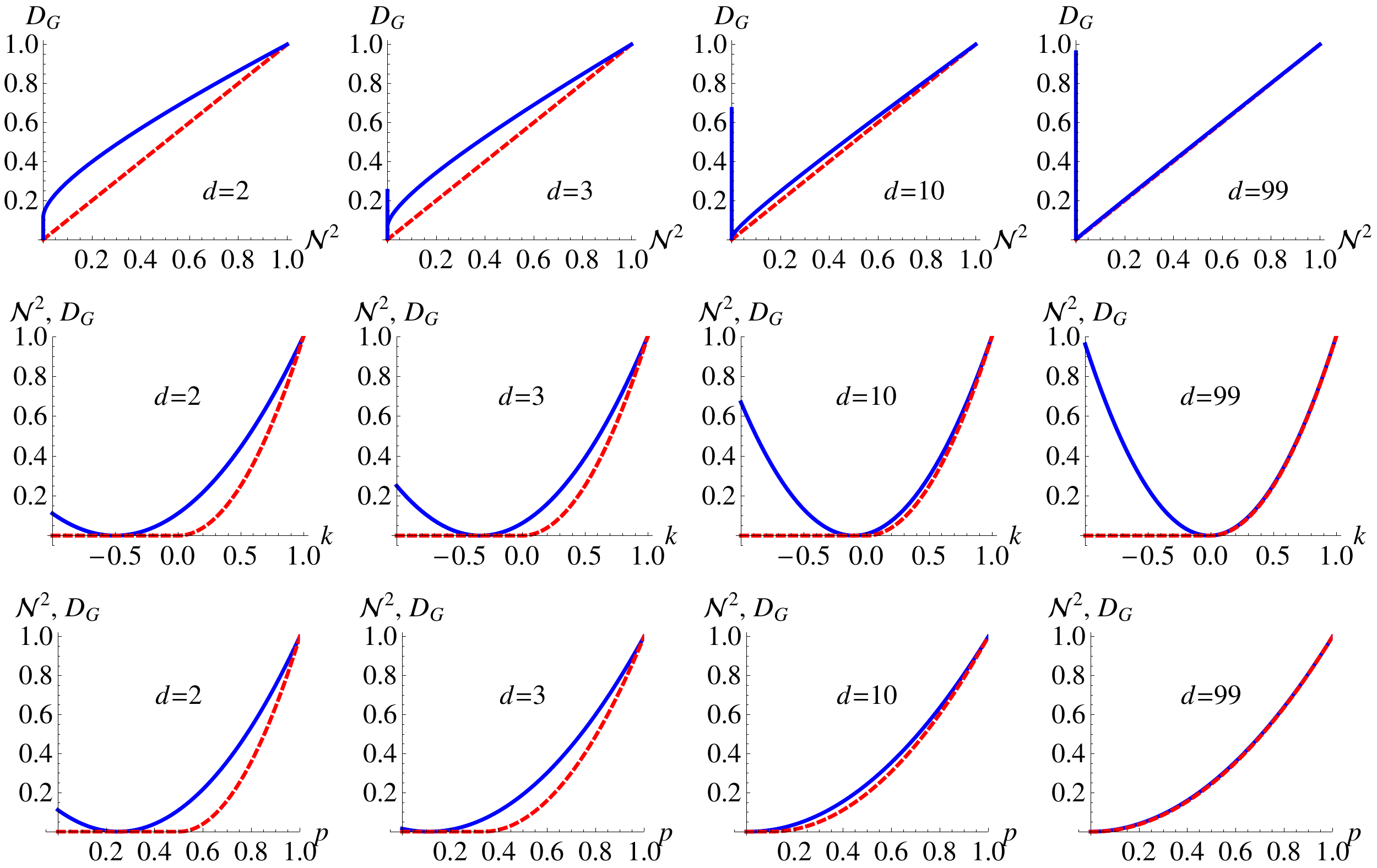}
\caption{(Color online) Top row: Geometric quantum discord versus squared negativity (solid blue line) [Eq.~(\ref{dgnweriso})] for $d \otimes d$ Werner states with dimensions $d=2,3,10,99$ (from left to right); the dashed red line of equation $D_G={\cal N}^2$ is just a guide to the eye. The corresponding plots for isotropic states are identical, apart from the extra vertical branch at ${\cal N}=0$ which is absent in those cases.
Middle row: $D_G$ (solid blue line) and ${\cal N}^2$ (red dashed line) for $d \otimes d$ Werner states \eqref{swer} plotted as a function of the parameter $k \in [-1,1]$.
Bottom row: $D_G$ (solid blue line) and ${\cal N}^2$ (red dashed line) for $d \otimes d$ isotropic states \eqref{siso} plotted as a function of the parameter $p \in [0,1]$.
All the quantities plotted are dimensionless. }\label{figweriso}
\end{figure*}

Coming back to our measures of choice in this work, geometric discord and squared negativity, we can visualize their interplay on pure two-qudit states with increasing $d$. We have generated a large ensemble of two-qudit states up to $d=7$ with random Schmidt coefficients. At fixed negativity, the geometric discord displays both upper and lower bounds. The upper bounds are multi-branched, with an increasing number of nodes appearing with increasing $d$. The lower bounds are  regular curves lying in general strictly above the bisectrix for any $d>2$, with $D_G = {\cal N}^2$ occurring only at the extremal points where both vanish (on factorized states) or both reach the maximum (on maximally entangled states).
We find that, for any $d$, the pure two-qudit states that achieve the minimum geometric discord at fixed negativity (green curve in Fig.~\ref{figqudits}) have a peculiar distribution of Schmidt coefficients:
\begin{eqnarray}\label{saturi}
\alpha_0&=&\sin^2\theta\,, \\
\alpha_i &=& \frac{\cos^2\theta}{d-1}\ \ \forall\, i=1,\ldots,d-1\,, \nonumber
\end{eqnarray} with $\arccos\sqrt{(d-1)/d} \le \theta \le \pi/2$. Since this is true for every pure state in the special case $d=2$, this is a further proof that on two-qubit pure states $D_G$ equals ${\cal N}^2$ as observed in the previous Section.
In general, the lower bound on $D_G$ at fixed ${\cal N}$ as saturated by the states of \eq{saturi} is given by
\begin{eqnarray}\label{lboundd}
D_G^{low}({\cal N}) &=& \left[2 (d-R-1)+(d-2) (d-1) \mathcal{N}\right] \nonumber \\
 & \times & \left[2 \left((d-1)^2+R\right)-(d-2) (d-1) \mathcal{N}\right] \\
 & \times & \left[(d-1)^2 d^2\right]^{-1}\,, \nonumber
\end{eqnarray}
with $R=\sqrt{(d-1)^2 (1-\mathcal{N}) [1+(d-1) \mathcal{N}]}$.

\subsection{Werner and isotropic $\boldsymbol{d \otimes d}$ states}

The ordering relationship between geometric discord and squared negativity can be further extended rigorously to two special classes of {\it mixed} $d \otimes d$ highly symmetric states, namely the Werner states \cite{Werner89} and the isotropic states \cite{isot}. We recall that for both families of states the PPT criterion is necessary and sufficient for separability \cite{voll}.

The Werner states in arbitrary $d$ dimension take the form \cite{Werner89}
\begin{eqnarray}\label{swer}
\rho_w=\frac{d+k}{d^3-d}\mathbb{I}_{d}+\frac{-d k -1}{d^3-d}| \Phi\rangle\langle \Phi|,
\end{eqnarray}
where $\ket\Phi =\sum_{i,j=0}^{d-1}(|i j\rangle+|j i\rangle)$ and $k\in [-1,1]$ with $0< k \leq 1$ for entangled states. The geometric discord calculated in \cite{luofu} and then normalized is
\begin{eqnarray}
D_G(\rho_w)=\frac{(d k+1)^2}{(d+1)^2},
\end{eqnarray}
while after simple algebra we obtain the following expression for the (normalized) negativity:
\begin{eqnarray}
{\cal N}(\rho_w)=\max\left\{0, k \right\}.
\end{eqnarray}

The isotropic states can be instead defined as \cite{isot}
\begin{eqnarray} \label{siso}
\rho_i=\frac{1-p}{d^2-1}\mathbb{I}_{d}+\frac{d^2 p -1}{d^2-1}| \Psi\rangle\langle \Psi|,
\end{eqnarray}
where $\ket\Psi =\frac{1}{\sqrt{d}}\sum_{i=0}^{d-1}|i i\rangle$ and $p\in [0,1]$, with the states being entangled for $p > 1/d$.
In such a case the (normalized) geometric discord  \cite{luofu} is
\begin{eqnarray}
D_G(\rho_i)=\frac{(d^2 p-1)^2}{(d^2-1)^2},
\end{eqnarray}
and the negativity is given by \cite{lee}
\begin{eqnarray}
{\cal N}(\rho_i)= \max\left\{0,\frac{d p-1}{d-1}\right\}.
\end{eqnarray}

Interestingly, for both classes of states (in the nontrivial region of parameters where they are entangled) some straightforward algebra shows that the simple relationship
\begin{equation}
D_G(\rho_{w,i})=\left[\frac{1+d\  {\cal N}(\rho_{w,i})}{1+d}\right]^2 \ge {\cal N}^2(\rho_{w,i})
\label{dgnweriso}
\end{equation}
holds, thus establishing once again the desired hierarchy.

We notice a radical difference between quantum correlations of Werner and isotropic states in the region in which they are separable. Namely, for Werner states the quantum correlations measured by the geometric discord can grow approaching the maximum (which is one in normalized units) even without entanglement, with increasing dimension $d$, so that $\lim_{d \rightarrow \infty} D_G(\rho_w) = k^2$ in the full range $k \in [-1,1]$, half of which contains separable states. Therefore Werner states with high dimension and $k \rightarrow -1$ are examples of highly mixed, completely separable states whose quantum correlations are asymptotically as big as those of pure maximally entangled states, as predicted in \cite{genio} (see also \cite{chita}). On the other hand, for isotropic states, with increasing $d$ the separability region ($0\leq p \leq 1/d$) just shrinks to zero, meaning that in such a case the geometric discord just converges to the squared negativity in the full parameter range, with no significant signatures of quantum correlations exhibited in absence of entanglement. Note that the two families of states are instead completely equivalent in the limiting case $d=2$ (upon identifying $k=2p-1$). The interplay between $D_G$ and ${\cal N}^2$ for Werner and isotropic states of varying dimension is illustrated in Fig.~\ref{figweriso}.

\subsection{Generic $\boldsymbol{d \otimes d'}$ states}

Encouraged by the previous results, we now wish to test the validity of the  inequality $D_G\geq{\cal N}^2$ for generic mixed states of arbitrary $d \otimes d'$ dimensional systems. Specifically, we run a numerical exploration of the $D_G$ versus ${\cal N}^2$ plane for randomly generated mixed states of $2 \otimes 3$ systems. In this case, the geometric discord can be computed according to the prescription of Ref.~\cite{rau}, while the negativity still captures all entanglement potentially present in the states \cite{PPT}. Remarkably, based on extensive numerical evidence (see Fig.~\ref{duepertre}), we find that the hierarchy between geometric discord and squared negativity holds as well for arbitrary states of a qubit and a qutrit. This finding, in addition to the results of the previous Sections, motivates us to conjecture that $D_G\geq{\cal N}^2$ might be a universal ordering relationship for arbitrary $d \otimes d'$ dimensional systems. A general proof of this statement would be very valuable, and an interesting, related open question concerns investigating the role of bound entanglement in higher dimensions  and its interplay (not captured by the negativity) with geometric measures of quantum correlations.

\begin{figure}[t]
\includegraphics[width=8.2cm]{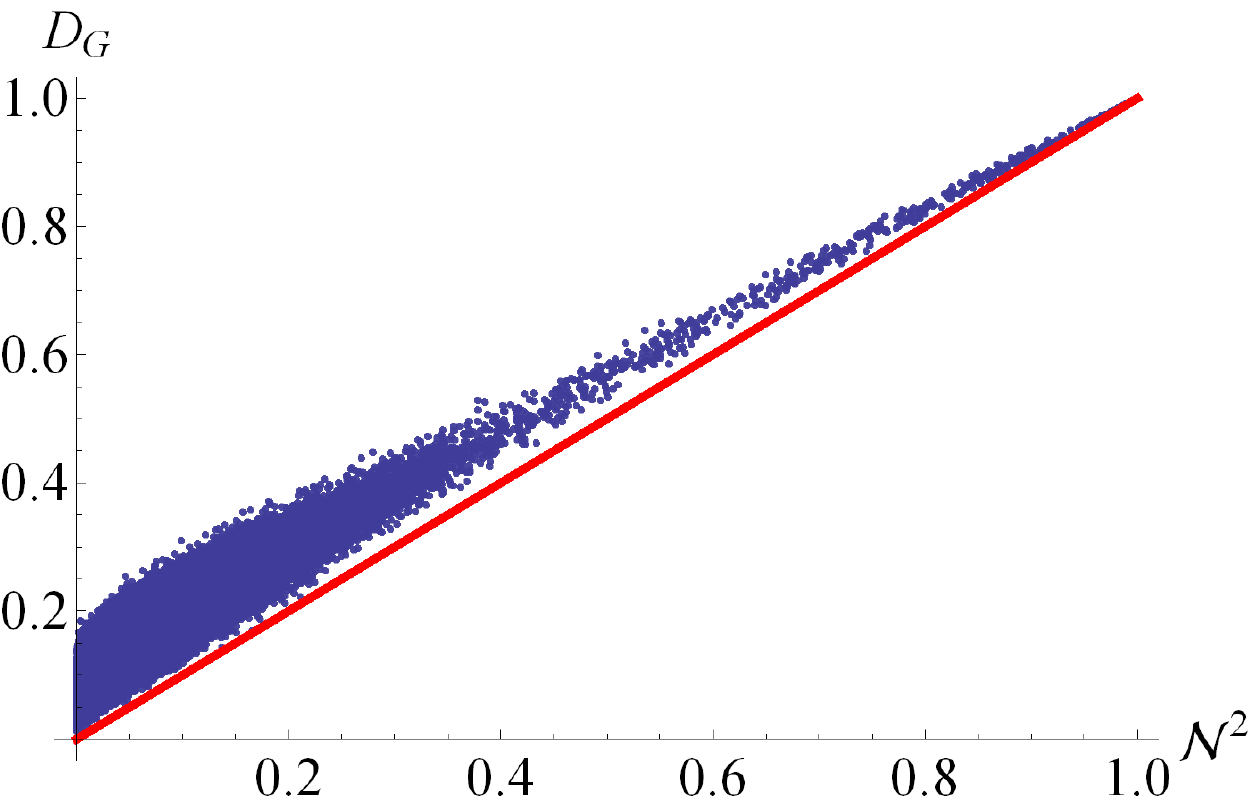}
\caption{Geometric discord versus squared negativity for $2\times 10^5$ mixed states of $2\otimes 3$ systems, randomly generated by using the \emph{Mathematica} package available in \cite{pacco}. All the quantities plotted are dimensionless.}
\label{duepertre}
\end{figure}

\section{Concluding remarks}\label{SecConcl}
 We have presented a qualitative and quantitative study of entanglement and general quantum correlations for arbitrary two-qubit states and for relevant instances of higher-dimensional states.

 First, we identified a computable measure of entanglement, the squared negativity ${\cal N}^2$ \cite{vidwer}, and proved that it is always majorized by a compatible measure of quantum correlations, the geometric discord $D_G$ \cite{dakic}, in the case of generic two-qubit states. The inequality is saturated for pure states.  We also provided numerical evidence that the squared negativity is still majorized by a tight lower bound $Q$ to the geometric discord, recently proposed as observable measure of quantum correlations \cite{pirla}. Thus, the chain $D_G\geq Q \geq {\cal N}^2$ holds for two-qubit states.

   Then, we explored the pattern of the plane $D_G$ vs ${\cal N}^2$, identifying the classes of two-qubit states with maximal geometric discord at fixed negativity. In particular, the bound is reached by a family of $X$ states given in \eq{X2}.  Remarkably, for separable states the upper bound accomodates  a fully asymmetric state, i.e. a state becoming a zero-discord classical-quantum state upon swapping of the subsystems.

Finally, we extended our analysis to arbitrary $d\otimes d'$ systems. For two-qudit pure states, we found that the hierarchy between geometric discord and squared negativity still holds rigorously. We characterized the states with minimal $D_G$ at fixed ${\cal N}^2$: they present an elegant parametrization of the distribution of their Schmidt coefficients, allowing to express analytically the lower bound in the $D_G$ vs ${\cal N}^2$ plane for any $d$ as in \eq{lboundd}. In the mixed-state case, the inequality is still valid for Werner and isotropic $d \otimes d$ states, for which $D_G$ is a simple function of the negativity for each dimension $d$. We further provided numerical evidence supporting the validity of the hierarchy between geometric discord and squared negativity for general mixed states of $2 \otimes 3$ systems. In all the instances analyzed in this paper, $D_G$ and ${\cal N}$ were computable in closed form, and were always found to obey the ordering relationship $D_G \ge {\cal N}^2$. We thus conjecture its validity on arbitrary bipartite states of any dimension, leaving open at the present stage the task of providing a rigorous general proof (or a counterexample) to our claim.

 Our results agree with the intuitive prediction that general quantum correlations should be somehow related to entanglement and definitely incorporate it \cite{nuovogenio}. Geometric discord \cite{dakic} (or its lower bound $Q$ \cite{pirla}) and negativity \cite{vidwer} are two computable, observable and experimentally friendly measures of quantum correlations whose interplay, explored in this paper, is important to get a quantitative grip on the performance of several quantum information protocols, ranging from quantum computation to quantum metrology and state discrimination \cite{dattacommun}. Understanding the nature of nonclassical correlations and their role in determining advantages over fully classical scenarios is a central issue in quantum information processing and communication \cite{merali}. On a more fundamental level, our findings suggest that nonclassical correlations measured by geometric discord could be regarded as a more general feature that somehow incorporates entanglement and state mixedness, following the intuition advanced in \cite{genio}.  Encouraging preliminary evidence that (the lower bound $Q$ to)  the geometric discord can be employed to characterize the dynamics of quantum correlations in open systems, and possibly other relevant features of open systems themselves (e.g. non-Markovianity \cite{nonM}), has been recently presented in \cite{pirla}. In this respect, the ordering relations we found suggest that entanglement and general quantum correlations as well can be both interrelated to such properties of open systems.
 Encouraged by the hierarchy pointed out in this paper, we believe it becomes even more meaningful to keep searching for simple but universal, physically motivated and mathematically accessible, unifying measures of  ``quantumness''  of the correlations, along the spirit of Refs.~\cite{modi,genio,bruss,nuovogenio}.

\acknowledgments{We acknowledge fruitful discussions with Shunlong Luo, Marco Piani, and Karol \.Zyczkowski. We thank the University of Nottingham for financial support through an Early Career Research and Knowledge Transfer Award and a Graduate School Travel Prize Award.}

\end{document}